\documentclass{article}
\usepackage{graphicx} 
\usepackage[utf8]{inputenc}
\usepackage{geometry}
\usepackage{amsmath,amsfonts,amssymb,amsthm, bm}
\usepackage{wrapfig}
\usepackage{mathtools}
\usepackage{commath}
\usepackage[sc,osf]{mathpazo}
\usepackage{bbm}
\usepackage{wrapfig}
\usepackage{subcaption}
\usepackage{listings}
\usepackage{color,soul}
\usepackage{multirow}
\usepackage{float}
\usepackage{pifont}
\usepackage{url}%
\usepackage{xifthen}
\usepackage{hyperref}
\usepackage{url}
\usepackage{fancyhdr}
\usepackage{threeparttable}
\usepackage{framed}
\usepackage{diagbox}
\usepackage{makecell}
\usepackage{algorithm}
\usepackage[noend]{algpseudocode}
\usepackage{float}
\usepackage[table]{xcolor}

\newcommand{\stot}{S_{tot}}
\newcommand{\stake}{S}
\newcommand{\numvalidators}{N}
\newcommand{\palg}{\mathcal{P}_{alg}}
\newcommand{\psoc}{\mathcal{P}_{soc}}
\newcommand{\pecon}{\mathcal{P}_{econ}}
\newcommand{\onethird}{\frac{1}{3}}
\newcommand{\ws}{T_{ws}}
\newcommand{\tvl}{TVL}

\newcommand{\reversion}{T_{rev}}
\newcommand{\advfork}{\mathcal{C}_{adv}}
\newcommand{\txset}{\Gamma}
\newcommand{\hybtxset}{\Gamma_{hyb}}
\newcommand{\hybsecconftxset}{\Gamma_{hyb, \secconf}}
\newcommand{\secconf}{\Pi_{sec}}
\newcommand{\secbri}{\Pi_{sec-bridge}}
\newcommand{\crw}{T_{cr}}
\newcommand{\setclient}{\Theta}
\newcommand{\client}{tr}
\newcommand{\clienthybtxset}{\Gamma_{\client}}
\newcommand{\hybuninstxset}{\Gamma_{un-ins}}
\newcommand{\sins}{S_{ins}}
\newcommand{\sinsavail}{S_{ins, avail}}
\newcommand{\inspur}{\Pi_{ins-pur}}
\newcommand{\inssec}{\Pi_{ins-sec}}
\newcommand{\clientinspur}{u}

\newtheorem{lemma}{Lemma}

\title{STAKESURE: Proof of Stake Mechanisms with Strong Cryptoeconomic Safety.}
\author{
  Soubhik Deb\\
  \texttt{EigenLabs}\\
  \texttt{soubhik@eigenlabs.org}
  \and
  Robert Raynor\\
  \texttt{EigenLabs}\\
  \texttt{rraynor@eigenlabs.org}
  \and 
  Sreeram Kannan\\
  \texttt{EigenLabs}\\
  \texttt{sreeram@eigenlabs.org}
}
\date{}

\begin{document}

\maketitle
\begin{abstract}
    As of June 15, 2023, Ethererum, which is a Proof-of-Stake (PoS) blockchain \cite{buterin2022proof} has around 410 Billion USD in total assets on chain (popularly referred to as total-value-locked, TVL) but has only ~33 Billion USD worth of ETH staked in securing the underlying consensus of the chain \cite{ultrasound}. A preliminary analysis might suggest that as the amount staked is far less (11x less) than the value secured, the Ethereum blockchain is insecure and “over-leveraged” in a purely cryptoeconomic sense. In this work, we investigate how Ethereum, or, more generally, any PoS blockchain can be made secure despite this apparent imbalance. Towards that end, we attempt to formalize a model for analyzing the cryptoeconomic safety of PoS blockchain, which separately analyzes the cost-of-corruption, the cost incurred by an attacker, and the profit-from-corruption, the profit gained by an attacker. We derive sharper bounds on profit-from-corruption, as well as new confirmation rules that significantly decrease this upper-bound. We evaluate cost-of-corruption and profit-from-corruption only from the perspective of attacking safety. Finally, we present a new “insurance” mechanism, STAKESURE, for allocating the slashed funds in a PoS system, that has several highly desirable properties: solving common information problem in existing blockchains, creating a mechanism for provably safe bridging, and providing the first sharp solution for automatically adjusting how much economic security is sufficient in a PoS system. Finally, we show that the system satisfies a notion of strong cryptoeconomic safety, which guarantees that no honest transactor ever loses money, and creates a closed system of Karma, which not only ensures that the attacker suffers a loss of funds but also that the harmed parties are sufficiently compensated. 
\end{abstract}

\section{Introduction}
\label{sec:intro}
The popular website ultrasound.money used to track various metrics on the Ethereum blockchains, has a special metric called the “security ratio” \cite{ultrasound}.  This essentially represents the total-value-locked (TVL) into all the assets on Ethereum divided by the total stake that is securing those assets in the underlying consensus of Ethereum. This security ratio is a whooping 11x as of June 1, 2023. At first guess, one would expect if this security ratio to be greater than 1x, then the blockchain is insecure, .  One can make identical observations about many other Proof-of-Stake (PoS) blockchains. Naturally, the question that one would ask is how is it possible that even though the amount of stake that is securing the consensus of such PoS blockchains is significantly less than the amount it is securing in the first place, these PoS blockchains are still safe and live? Is it merely that the blockchains have not yet been attacked, or is there a more fundamental reason, why the blockchain *cannot* be attacked. 

To understand this puzzle, we define two core economic metrics of PoS blockchains:
\begin{itemize}
    \item \textbf{Cost-of-corruption.} The minimum cost / loss that needs to be incurred by any adversary to mount any successful attack against the system,
    \item \textbf{Profit-from-corruption.} The maximum profit that an adversary can extract by being able to successfully attack the system.
\end{itemize}
If we have these two metrics at hand, then it is clear that the system is secure as long as 
\begin{equation}
    \underline{\textbf{Cryptoeconomic Safety:}} \quad \text{cost-of-corruption} > \text{profit-from-corruption}.
\end{equation}
As a system designer for a PoS blockchain, the goal would be to increase the cost-of-corruption as much as possible while substantially reducing the profit-from-corruption.

PoS blockchains like Ethereum have increased their cost-of-corruption by incorporating slashing into their mechanism design. Slashing offers a means to economically penalize any particular validator in a targeted manner for not taking a blockchain-compliant action. This is done by taking away some or all of the validator’s stake, without imposing externalities on other validators who are behaving according to the rules of conduct set for participating in a blockchain. Slashing is unique to PoS blockchains because it requires the ability for the blockchain to enforce the penalty. Such enforcement is clearly infeasible in proof-of-work blockchains, where it would be analogous to burning the mining hardware used by misbehaving validators. 

Decreasing profit-from-corruption is much subtler. This requires identifying the type of transactions whose confirmation can be used by an adversary to drain value out of the blockchain. Such transactions are those which require some follow-up off-chain action on the part of some client. Examples are centralized exchanges (CEXs) and bridges. The risk of executing the corresponding off-chain actions of such transactions too soon, even if they are finalized, is that a sufficient number of adversarial validators can engage in double-signing to build a separate finalized fork (despite the negative incentive of slashing) consisting of double-spending transaction while attempting to sow confusion among the clients in the ecosystem about which fork to be considered as “canonical”. Minimizing profit-from-corruption would require coming up with a confirmation rule on when to execute such transactions.

While existing staking and slashing systems burn the funds when there is an attack, thus causing the adversary to take a hit, there is no compensation to the harmed parties, i.e, honest transactors. This makes the blockchain system not have a closed loop of Karma. We also define a stronger notion of cryptoeconomic safety, which ensures that there is no unintended loss of funds to a honest party: 

\begin{itemize}
    \item[] \underline{\textbf{Strong cryptoeconomic safety:}} No honest user of the system suffers any loss of funds. 
\end{itemize}
We note that even blockchains that satisfy cryptoeconomic safety may not satisfy the notion of strong cryptoeconomic safety. While cryptoeconomic safety ensures that there is no incentive for an adversary to attack, a malicious adversary may still go ahead and attack the system which will lead to honest users in the system suffering without recourse. In contrast, in a system with strong cryptoeconomic safety, this can never happen. 

\subsection{Contributions and organization}
The contributions of this article are:
\begin{enumerate}
    \item Mathematical model for cryptoeconomic safety in sec.~\ref{sec:modeling}.
    \item In sec.~\ref{sec:coc}, we calculate cost-of-corruption from the persective of attacking safety. We show how slashing can help in increasing this cost as compared to other economic mechanisms commonly used. 
    \item In sec.~\ref{sec:pfc}, we show how to obtain sharp bounds on profit-from-corruption from the perspective of attacking the safety of a PoS blockchain. This involve sharply defining what family of transactions should be considered for computing the profit-from-corruption. Then, we proposed a confirmation rule for clients on when to consider such transactions confirmed. 
    \item In sec.~\ref{sec:insurance}, we propose a new economic mechanism for PoS blockchains, called as STAKESURE, which can be used for obtaining staking insurance for transactors thus guaranteeing strong cryptoeconomic safety, which ensures that all transactors get unconditional protection from attacks on safety.

\end{enumerate}

\section{Modeling}
\label{sec:modeling}

\subsection{Modeling consensus}
At a high-level, consensus in blockchain is a social process among a distributed set of validators to come to an agreement on what should be the next set of updates to the shared state. However, some aspects of consensus feature “algorithmic-attributability”, that is, certain faults can be detected via algorithms running in softwares without requiring any human-in-the-loop. Examples of such faults are double-signing, verification of execution done by a zero-knowledge (ZK) rollup and an optimistic rollup (if data is available), etc. Furthermore, in such faults, the adversary can be identified in an objective manner and can be penalized. Obviously, the algorithmic-attributable nature of such faults offers an opportunity to take advantage of superior computational power of machines to come to agreement on what should be the next updates to the shared state very rapidly. 

On the other hand, there are aspects of consensus which are “socially-attributable” in nature. In such cases, faults become apparent at a much larger timescale and any resolution of those faults require extensive social coordination among the clients in the blockchain ecosystem and not just the validators who are running the full node softwares.  Examples of such faults are loss of ledger growth and censorship. In social-attributable faults, social coordinations feature usage of social media to have discussion and lively debates on any faults/attacks arising in the blockchain and its peripheral infrastructure. This usage of social media helps in building a schelling point to coordinate the actions of honest clients in determining what fork of the blockchain should be considered as the “canonical” one.

Traditionally, while analyzing the security of PoS blockchains, it has been assumed to consist of only a consensus protocol where a group of $\numvalidators$ validators are each running full node software to come to consensus on what new block of transactions should be added to the ledger. Security analysis of the underlying consensus protocol would then provide the bounds under which the PoS blockchain would be considered to be safe and live. However, what has been woefully neglected in these security analysis is the supporting set of clients who are running light client software. Whenever there is a socially-attributable fault such as liveness attack, clients running light client software or full node software socially coordinate with each other via social media to resolve the fault. 

In order to do a comprehensive understanding of cryptoeconomics of PoS blockchains, we categorize coordination mechanisms underlying a blockchain into two parts:
\begin{itemize}
    \item \textbf{Algorithmic consensus ($\palg$).} This consists of special clients called as validator, each of whom are running full node software to participate in an algorithmic consensus protocol like a BFT protocol or a longest-chain protocol to reach resolutions on what should be the next update to the shared state in the presence of algorithmically-attributable faults like safety attacks. For this paper, we assume that $\palg$ is a BFT protocol with single-slot finality and adversarial threshold of $\frac{1}{3}$.
    \item \textbf{Social consensus ($\psoc$).} This consists of clients running full node or light client software to engage in social coordination via any channels (like social media) in order to resolve faults of “socially-attributable” in nature. Social consensus happens at a much longer time horizon than algorithmic consensus \cite{buterin2022proof}.
\end{itemize}
On top of coordination mechanism, we represent the economic aspects of a PoS blockchain involving positive and negative incentives via $\pecon$. 

\subsection{Models for analyzing security}
In order to understand why PoS blockchains like Ethereum are secure, we first need a model under which we will pursue our analysis. Two of the most popular models for analyzing PoS blockchains, the Byzantine model and the game-theoretic equilibrium model, fail to capture economic aspects of real-world attacks. In this section, we discuss these existing models to understand their shortcomings, and present a third model – what we call the Corruption-Analysis Model – based on separately evaluating the bounds on the minimum cost that has to be incurred and the maximum profit that can be extracted from corrupting the blockchain. Despite its ability to model large swathes of attacks, the Corruption-Analysis Model has not yet been used for analyzing any blockchain. 

\subsubsection{Existing models}
In this section, we provide a brief description of Byzantine and game-theoretic equilibrium models and their shortcomings.
\newline

\noindent \underline{Byzantine model} \newline
The Byzantine model stipulates that at most a certain fraction ($\beta$) of validators can deviate from the $\palg$-prescribed actions and pursue any action of their choice, while the rest of the validators remain compliant with $\palg$. Proving that $\palg$ is resilient against a whole space of Byzantine actions that an adversarial validator can take is a non-trivial problem.

For example, consider $\palg$ to be a longest-chain PoS consensus protocol where liveness is prioritized over safety. Early research on security of longest-chain consensus focused on showing security against only one specific attack – the private double-spend attack, where all Byzantine validators collude to build an alternative chain in private and then reveal it much later once it is longer than the original chain. The nothing-at-stake phenomenon, though, offers an opportunity to propose a lot of blocks using the same stake and to use independent randomness to increase the probability of constructing a longer private chain \cite{nas}. Only much later, extensive research was undertaken to show that certain constructions of longest-chain PoS consensus protocols can be made secure against all attacks for certain values of $\beta$ \cite{dembo2020everything, deb2021posat}. 

There is also a whole class of algorithmic consensus protocols called Byzantine Fault Tolerant (BFT) protocols, which prioritize safety over liveness. They also require assuming a Byzantine model for showing that, for an upper bound on $\beta$, these protocols are deterministically safe against any attack \cite{yin2018hotstuff, chan2020streamlet, buchman2018latest}. 

Although helpful, the Byzantine model doesn’t account for any economic incentives. From a behavioral perspective, $\beta$ fraction of these validators are completely adversarial in nature while (1-$\beta$ ) fraction are fully compliant with the specification of $\palg$. In contrast, a significant fraction of validators in a PoS blockchain may be motivated by economic gains and run modified versions of the validator software prescribed by $\palg$ that benefit their self interest rather than simply complying with the full specification of $\palg$. As a salient example, consider the case of Ethereum, where most validators today do not run the default PoS protocol but run the MEV-Boost modification, which results in additional rewards due to participation in a MEV auction market, rather than running the exact protocol specification.  
\newline 

\noindent \underline{Game-theoretic equilibrium model} \newline
The game-theoretic equilibrium model attempts to address the shortcoming of the Byzantine model by using solution concepts like the Nash equilibrium to study whether a rational validator has the economic incentives to follow a given strategy when all other validators are also following the same strategy. More explicitly, assuming everyone is rational, the model investigates two questions: 
\begin{enumerate}
    \item If every other validator is following the prescribed strategy from $\palg$, does it bring the most economic benefit for me to execute upon the same prescribed strategy? 
    \item If every other validator is executing the same deviating strategy, is it most incentive-compatible for me to still follow the prescribed strategy of $\palg$?
\end{enumerate}
Ideally, $\palg$ should be designed such that the answer to both questions is “yes”.

An inherent shortcoming in the game-theoretic equilibrium model is that it excludes the scenario where an exogenous agent might be influencing the behavior of validators. For example, an external agent can set up a bribe to incentivize rational validators to act in accordance with its prescribed strategy. Another limitation is that it assumes that each of the validators has the independent agency to make their own decisions on what strategy to follow based on their ideology or economic incentives. But this doesn’t capture the scenario where a group of validators collude to form cartels or when economies of scale encourage the creation of a centralized entity that essentially controls all staking validators. 
\newline

\noindent \underline{Corruption-Analysis Model} \newline
Several researchers have proposed the Corruption-Analysis Model for analyzing the security of any PoS blockchain, although none have performed a deeper analysis. The model starts by asking the following two questions:
\begin{enumerate}
    \item What is the minimum cost incurred by any adversary for successfully executing a safety or liveness attack on the blockchain?
    \item What is the maximum profit that an adversary can extract from successfully executing a safety or liveness attack on the blockchain?
\end{enumerate}

The adversary in question can be 
\begin{itemize}
    \item a validator that is deviating from PoS blockchain-prescribed strategy unilaterally, 
    \item a group of validators that are actively cooperating with one another to undermine the blockchain, or 
    \item an external adversary attempting to influence the decisions of many validators through some external action such as bribing. 
\end{itemize}

Computing the costs involved requires taking into consideration any cost incurred for bribes, any economic penalty incurred for executing upon a Byzantine strategy, and so on. Similarly, computing profit is all-encompassing, which counts any in-protocol reward obtained by successfully attacking the protocol, any capture of value  from the DApps sitting on top of the PoS blockchain, taking positions on protocol-related derivatives in secondary markets and profiteering from resultant volatility from the attack, and so on.

Comparing a lower-bound on the minimum cost for any adversary to mount an attack (cost-of-corruption) against an upper-bound on the maximum profit that an adversary can extract (profit-from-corruption) indicates when it is economically profitable to attack the protocol. This model has been used for analyzing Augur and Kleros \cite{Augur,kleros}. Thus, we consider a PoS blockchain to cryptoeconomically safe if
\begin{equation}
    \text{cost-of-corruption} > \text{profit-from-corruption}
\end{equation}

\section{Computing cost-of-corruption}
\label{sec:coc}

In this section, we will evaluate the cost-of-corruption for PoS blockchains. Towards that end, we will investigate two mechanism designs that are typically used as a safeguard against attacks to the blockchain: (1) token toxicity and (2) slashing. Let $\stot$ to be the total stake in a PoS blockchain. Using corruption-analysis model, we obtain the following table: 

\begin{table}[h!]
\begin{center}
\begin{tabular}{ |c|c|c| } 
\hline
\cellcolor{blue!25} Mechanism design & \cellcolor{blue!25} Cost-of-corruption \\
\hline
Token toxicity & $0$ \\ 
\hline
Slashing & $\frac{1}{3}\stot$ \\
\hline
\end{tabular}
\caption{Cost-of-corruption across different mechanism designs. Here, $\stot$ represents the total stake in the PoS blockchain.}
\end{center}
\end{table}

Clearly, incorporating slashing into the mechanism design of a PoS protocol substantially increases the cost-of-corruption that any adversary would incur. 

In the following sections, we present the full analysis. For this analysis, we will consider a PoS blockchain consisting of $\numvalidators$ rational validators (with no Byzantine or altruistic validators)  in $\palg$ who are participating in a BFT protocol with $\frac{1}{3}$ adversary threshold. Let’s assume, for simplicity of calculation, that each validator has deposited an equal amount of stake $\stake$. Let us also define $\stot= \numvalidators \times\stake$.

\subsection{Token toxicity}
Token toxicity states that if a protocol gets successfully attacked, the underlying token would lose value. Many think that stakers would view this toxicity as a threat against compromising the security of the protocol. Contrary to this belief, we show in this section that token toxicity is not potent enough to deter adversarial attacks in some typical scenarios. In fact, the cost incurred by adversaries to attack and corrupt the protocol, referred to as cost-of-corruption, under such scenarios is essentially zero. 

Consider the scenario where $\frac{1}{3}$rd of the validators have joined hands. If the total valuation of the token, in which stake has been deposited, strictly depends on the security of the protocol, then any attack on the safety of the protocol can drive down its total valuation to zero. Of course, in practice, it will not be driven down all the way to zero but to some smaller value. But to present the strongest possible case for the power of token toxicity, we will assume here that token toxicity works perfectly.  The cost-of-corruption for any attack on the protocol is the total amount of tokens held by the rational validators who are attacking the system, who must be willing to lose all that value.

We will now analyze the incentives for collusion and bribing in a PoS blockchain with token toxicity without slashing. Suppose that the external adversary sets up the bribe with the following conditions:
\begin{itemize}
    \item If a validator executes upon the strategy as dictated by the adversary but the attack on the protocol was not successful, then the validator gets a reward $B_1$ from the adversary.
    \item If a validator executes upon the strategy as dictated by the adversary and the attack on the protocol was successful, then the validator gets a reward $B_2$ from the adversary. 
\end{itemize}

We can draw the following payoff matrix for a validator who has deposited stake $\stake$, and $R$ is the reward from participating in the PoS protocol:
\begin{center}
\begin{tabular}{ |c|c|c|c| } 
\hline
 & Attack not successful & Attack successful \\
\hline
\parbox{4cm}{A validator not taking up the bribe and not deviating from the protocol} & $S+R$ & 0 \\ 
\hline
\parbox{4cm}{A validator agreeing to take up the bribe} & $S+B_1$ & $B_2$ \\
\hline
\end{tabular}
\end{center}

Suppose that the adversary sets the bribe payoff such that $B_1>R$ and $B_2>0$. In such a case, accepting bribes from the adversary gives higher payoff than any other strategy the validator can take irrespective of the strategy other validators are taking (the dominant strategy). If $\frac{1}{3}$rd  of other validators end up accepting the bribe, they can attack the security of the protocol (this is because we assume we are using a BFT protocol whose adversary threshold is $\frac{1}{3}$) . Now, even if the present validator does not take the bribe, the token would anyway lose its value due to token toxicity (top right cell in the matrix). Therefore, it is incentive-compatible for the validator to accept the $B_2$ bribe. If only a small fraction of validators accept the bribe, the token won’t lose value, but a validator can benefit from forgoing the reward R and instead get $B_1$(left column in the matrix). In case of a successful attack where $\frac{1}{3}$rd of validators have agreed to accept the bribe, the total cost incurred by the adversary in paying out the bribes is at least $\frac{N}{3} B_2$;  this is the cost-of-corruption. However, the only condition on $B_2$ is that it has to be greater than zero and hence, $B_2$ can be set close to zero which would imply cost-of-corruption is negligible. This attack is known  as “$P+\epsilon$” attack \cite{pplusepsilon}.

One way of summarizing this effect is that token toxicity is insufficient because the impact of bad actions are socialized: token toxicity depreciates the value of token completely and affects good and bad validators equally. On the other hand, the benefit of taking the bribe is privatized and limited to only those rational validators that actually take the bribe. There is no  one-to-one consequence only for those taking the bribe, that is, the system doesn’t not have a working version of "karma."

\subsection{Slashing}
Slashing is a way for PoS blockchains to economically penalize a validator or a group of validators for executing a strategy that is provably divergent from the given protocol specification.  Typically, to enact any form of slashing, each validator must have previously committed some minimum amount of stake as a collateral. 

We concern ourselves primarily with the study of slashing mechanisms for safety violations such as double-signing, rather than for liveness violations. We suggest this restriction for two reasons: (1) safety violations are fully attributable in some BFT-based PoS protocols, but liveness violations are not attributable in any protocol, and (2) safety violations are usually more serious than liveness violations, resulting in loss of user funds rather than users unable to issue transactions.

Suppose an external adversary is able to bribe at least $\onethird$ of the validators in the PoS blockchain. With each validator having put up a stake of $\stake$ and all slashed stake is burnt, we get the following payoff matrix: 
\begin{center}
\begin{tabular}{ |c|c|c|c| } 
\hline
 & Attack not successful & Attack successful \\
\hline
\parbox{4cm}{A validator not taking up the bribe and not deviating from the protocol} & $S+R$ & $\stake$ \\ 
\hline
\parbox{4cm}{A validator agreeing to take up the bribe} & $S+B_1$ & $B_2$ \\
\hline
\end{tabular}
\end{center}
Observe that with slashing, if the validator agrees to take up the bribe and the attack is not successful, then its stake $\stake$ gets slashed in the canonical fork (lower left cell in the matrix), which is in contrast with the previous bribing scenario with token toxicity. On the other hand, a validator would never lose its stake $\stake$ in the canonical fork even if the attack is successful (top right cell in the matrix).  If it requires $\onethird$ of the total stake $\stot$ to be corrupted for the attack to be successful, we have
\begin{equation}
    \text{Cost-of-corruption} = \frac{N}{3}\times\stake = \frac{1}{3}\stot,
\end{equation}
which is substantially greater than the cost-of-corruption without slashing.

\subsection{Capturing weak subjectivity}
Even if PoS blockchains incorporate slashing of those validators that engage in reorgs by double-signing, validators can engage in a form of nothing-at-stake attack with impunity. More explicitly, consider a group of validators whose stake amounted to more than $\frac{2}{3}$ of $\stot$ long time back in the past in $\palg$ but have already exited from staking since then. These group of validators can now build a fork starting from an old block and given they control more than $\frac{2}{3}$ of $\stot$, these fork will be finalized. Even though it is apparent to everyone that this new fork is adversarial, unfortunately, they can't be slashed in the canonical fork as these validators have already exited. This implies cost-of-corruption is zero. 

One way that PoS blockchains like Ethereum and Tendermint have got around this problem is by weak subjectivity mechanism \cite{weaksubjectivity}. Under this mechanism, if a newly revealed fork $\mathcal{C}_{adv}$ builds upon a block which is considered to be sufficiently old (older than the most recent weak subjectivity period $\ws$), it won’t be accepted objectively by anyone actively running even a light client software.

\section{Computing profit-from-corruption}
\label{sec:pfc}
In this section, we sharpen our estimate on profit-from-corruption for PoS blockchains by formally defining the set of transactions which are vulnerable to chain reorgs. For any transaction $tx$, we define value of that transaction as $v(tx)$. First we make the observation that an adversary can confuse the social consensus $\psoc$ only if two forks are released in a near-simultaneous manner. In practice, if the second fork is released beyond a certain period, then the social consensus can coordinate to unequivocally converge upon considering the first fork to be the canonical. This period is represented as $\reversion$.  Defining $\txset(t,t+\reversion)$ to be the set of transactions that are being executed within a time interval $[t, t+\reversion)$, it is evident that these transactions are vulnerable to a chain reorg. However, those transactions that are purely atomic on-chain in nature and have no corresponding off-chain action(s) are safe from state reversion due to chain reorg.  Therefore, we can make even sharper bounds on profit-from-corruption by realizing that only those transactions that have corresponding off-chain actions  
are vulnerable.Examples of such transactions are exchanging ETH for fiat via CEXs, bridging assets from chain $A$ to chain $B$, etc. We represent the set of all such transactions that happen in the time interval $[t, t+\reversion)$ as $\hybtxset(t,t+\reversion)$. One way to reduce the size of the set $\hybtxset(t,t+\reversion)$ is by having the transactors follow a new confirmation rule that a transaction is considered to be confirmed if no conflicting fork appears within the reversion period since the block containing that transaction was finalized. Table 2 summarizes the profit-from-corruption under various adversarial safety violations.

\begin{table}[h!]
\begin{center}
\begin{tabular}{ |c|c|c| } 
\hline
\cellcolor{blue!25} Bounds  & \cellcolor{blue!25} Profit-from-corruption \\
\hline\hline
\parbox{6cm}{Stealing all value locked} & $\tvl$ \\ 
\hline\hline
\parbox{6cm}{Reorg within reversion period} & $\sup_{t > 0} \sum_{tx \in \Gamma (t,t+\reversion)} v(tx)$ \\ 
\hline\hline
\parbox{6cm}{Reorg against hybrid transactions within reversion period} & $\sup_{t > 0} \sum_{tx \in \hybtxset (t,t+\reversion)} v(tx)$\\
\hline\hline
\parbox{6cm}{Reorg against hybrid transactions within reversion period and $\secconf$} & $\sup_{t > 0} \sum_{tx \in \hybtxset (t,t+\reversion) \setminus \hybsecconftxset(t,t+\reversion)}  v(tx)$\\
\hline\hline
\end{tabular}
\caption{$\reversion$ represents the smallest time $t$ such that if the adversary releases its fork to reorg a  block $t$ time units after the block was finalized, then the social consensus $\psoc$ will be able to converge without ambiguity that this is an adversarial fork. $\txset(t,t+\reversion)$ is the set of transactions that are being executed within a time interval $[t, t+\reversion)$. $\hybtxset(t,t+\reversion)$ is the set of hybrid transactions whose off-chain actions are being executed immediately after they were finalized at some point in time in the interval $[t, t+\reversion)$. $\hybsecconftxset(t,t+\reversion)$ is the set of hybrid transactions whose off-chain actions are being executed according to the proposed secure confirmation rule $\secconf$ after they were finalized at some point in time in the interval $[t, t+\reversion)$.}
\end{center}
\label{tab:pfc}
\end{table}
In following sections, we delve deeper into each of these safety violations.

\subsection{Naive analysis}
\subsubsection{Stealing all assets locked}
Naively, one would guess that an adversary might as well drain out all total-value-locked ($\tvl$) with incorrect state transition. Therefore, we have the following condition that must be satisfied for a PoS blockchain to be secure:

\begin{equation}
    \frac{1}{3} \times \stot > \tvl.
\end{equation}

Considering Ethereum as the example of a PoS blockchain, however, Ethereum’s $\tvl$ is in the order of hundreds of billions of USD whereas total stake $\stot$ in securing the $\palg$ of Ethereum is in the order of tens of billions of USD. Clearly, the above security condition is not satisfied at all yet Ethereum remains safe and live. A similar conclusion can be drawn for other PoS blockchains too. Therefore, we have to capture a stronger bound on profit-from-corruption.

Observe that since validators running full node software are actively participating in the algorithmic consensus $\palg$, they have the access to the underlying blob of transactions and therefore, can always verify the state transition corresponding to the proposed block. This implies that no profit can be extracted out by an adversary by fooling any of these honest validators through incorrect state transition. These honest validators would obviously also raise alert in the social consensus $\psoc$, thus, enabling all honest clients not running full node software to become aware of the incorrect state transition. Observe that incorrect state transition is an objectively attributable fault and so the honest clients in $\psoc$ will objectively reject the fork containing the incorrect state transition. Consequently, one might be motivated to say that profit-from-corruption is $0$. This results in following {\bf erroneous security claim} for a PoS blockchain: that the amount of stake is not a consideration for determining the safety of a blockchain.

The above reasoning on profit-from-corruption doesn’t consider attacks via chain reorgs: adversarial validators can engage in reorging the finalized blocks to include double-spending transactions while also causing confusion around which fork to be considered “canonical”. While computing profit-from-corruption, we need to include the profits made from double-spending.

\subsection{Capturing reorgs: Reversion period}
We note that the weak subjectivity period is defined to be in the range of weeks or months \cite{weaksubjectivitycalc}. So if a competing block is released within a weak subjectivity period $\ws$ of a target block, then the validators can be slashed. In case that there are two competing forks, social consensus is needed to determine which fork gets built upon. The natural candidate for the fork to build upon (colloquially called as canonical fork) is the fork that gets released earlier. Of course if two forks are released in a near-simultaneous manner, then it is not possible to determine, even socially, which fork is the earlier \cite{vitaliksocialconsensus}. We call this time resolution as the reversion time - or in other words, if a block is produced and finalized, without having a competing fork beyond the reversion time, even if a competitive block is created adversarially and released beyond the reversion time, the social consensus $\psoc$ will ensure that it is the earlier block that gets built upon. We represent the duration of reversion period by $\reversion$. See fig.~\ref{fig:reversion-period} for illustration. The duration of the reversion period $\reversion$ can be much smaller than the duration of the weak subjectivity period $\ws$. 

From a game theoretic viewpoint, what would be ideal is that by the time social consensus $\psoc$ is able to converge on what fork should be considered “canonical”, stake of “enough” of the adversarial validators in $\palg$ involved in building the adversarial fork should still be active in the canonical fork. This will ensure that the adversarial validators can be penalized by slashing. What will be considered “enough” is that the total stake of these adversarial validators still active in the canonical fork should amount to more than $\frac{1}{3}$ of total stake $\stot$ currently active in the canonical fork. This is because, considering $\palg$ to be a BFT protocol with $\onethird$ adversarial threshold, adversary needs to control more than $\frac{1}{3}$ of $\stot$ to be able to do a successful equivocation.

Let $T_0$ be the time when the most recent finalized block in the canonical fork was observed to have been finalized. Table 3 illustrates the differences among finalization period $T_{fin}$, reversion period $T_{rev}$ and weak subjectivity period $\ws$.

\begin{table}
\begin{center}
\begin{tabular}{ |c|c|c|c| } 
\hline
\cellcolor{blue!25} \parbox{3.5cm}{When is the adversarial fork revealed?}  & \cellcolor{blue!25} Definition & \cellcolor{blue!25} \parbox{2.5cm}{Cost on social consensus?} & \cellcolor{blue!25} Slashability?
 \\
\hline\hline
$t \in [T_0, T_0+T_{fin})$ & \parbox{6cm}{Assuming single-slot finality and breakdown of underlying honest majority assumption, adversarial validators can engage in not giving their attestations to proposed blocks so as to algorithmically prevent any block from finalizing in the algorithmic consensus $\palg$. $T_{fin}$ is defined as the time when the tip of at least one of the fork finalizes as per the specification of $\palg$ since the last finalization of the common ancestor block.} & Infinite & \parbox{3cm}{No
(since no double signing needed for preventing blocks from getting finalized)}
\\ 
\hline\hline
$t \in [T_0+T_{fin}, T_0 + T_{fin} + T_{rev})$
& \parbox{6cm}{Suppose the tip of a fork has finalized at time $T_{fin}$. If there is a breakdown of the underlying honest majority assumption and an adversarial fork is publicly revealed to be finalized at time t, it is socially impossible for an honest validator to know without any ambiguity which fork is the “canonical” one.} & Infinite & \parbox{3cm}{Yes 
(double signing violation as per algorithmic consensus $\palg$)
}
 \\ 
\hline\hline
$t \in [T_0 + T_{fin} + T_{rev}, T_0+T_{ws})$ & \parbox{6cm}{Suppose the tip of a fork has finalized at time $T_{fin}$. If there is a breakdown of the underlying honest majority assumption and an adversarial fork is publicly revealed to be finalized at time $t$, the honest clients in the ecosystem can employ social consensus to converge on what is the canonical fork without any ambiguity.} & Finite & \parbox{3cm}{Yes 
(double signing violation as per algorithmic consensus $\palg$)}
\\ 
\hline\hline
$t \in [T_0+T_{ws}, \infty)$ & \parbox{6cm}{If an adversarial fork is publicly revealed to be finalized at time $t$, the honest validators can just trivially use their local information to discard the adversarial fork.} & None & \parbox{3cm}{No
(no stake left to be slashed under long-range attack)}\\
\hline\hline
\end{tabular}
\caption{Summary of different resolution periods.}
\end{center}
\end{table}

Suppose that a client in $\psoc$ running a light client waits until the end of the reversion period after it hears about the finalization of the block containing a certain transaction before the client considers the transaction to be confirmed. If the adversary publicly reveals a new fork $\advfork$ after the reversion period is over, the social consensus would be able to detect the adversarial fork. Consequently, no honest client would respect this fork $\advfork$ and no honest validator would propose or attest blocks on top of that fork $\advfork$ and therefore, the client would be safe from double-spending.  Therefore, profit-from-corruption should only include the total value that transacted within the reversion period. Hence, the security condition that needs to be satisfied for any PoS blockchain to be considered secure is given by: 
\begin{equation}
    \frac{1}{3} \times \stot > \sup_{t > 0} \sum_{tx \in \txset (t,t+\reversion)} v(tx). 
\end{equation}

\begin{figure}[!b]
     \centering
     \begin{subfigure}[b]{0.55\textwidth}
         \centering
         \includegraphics[width=\textwidth]{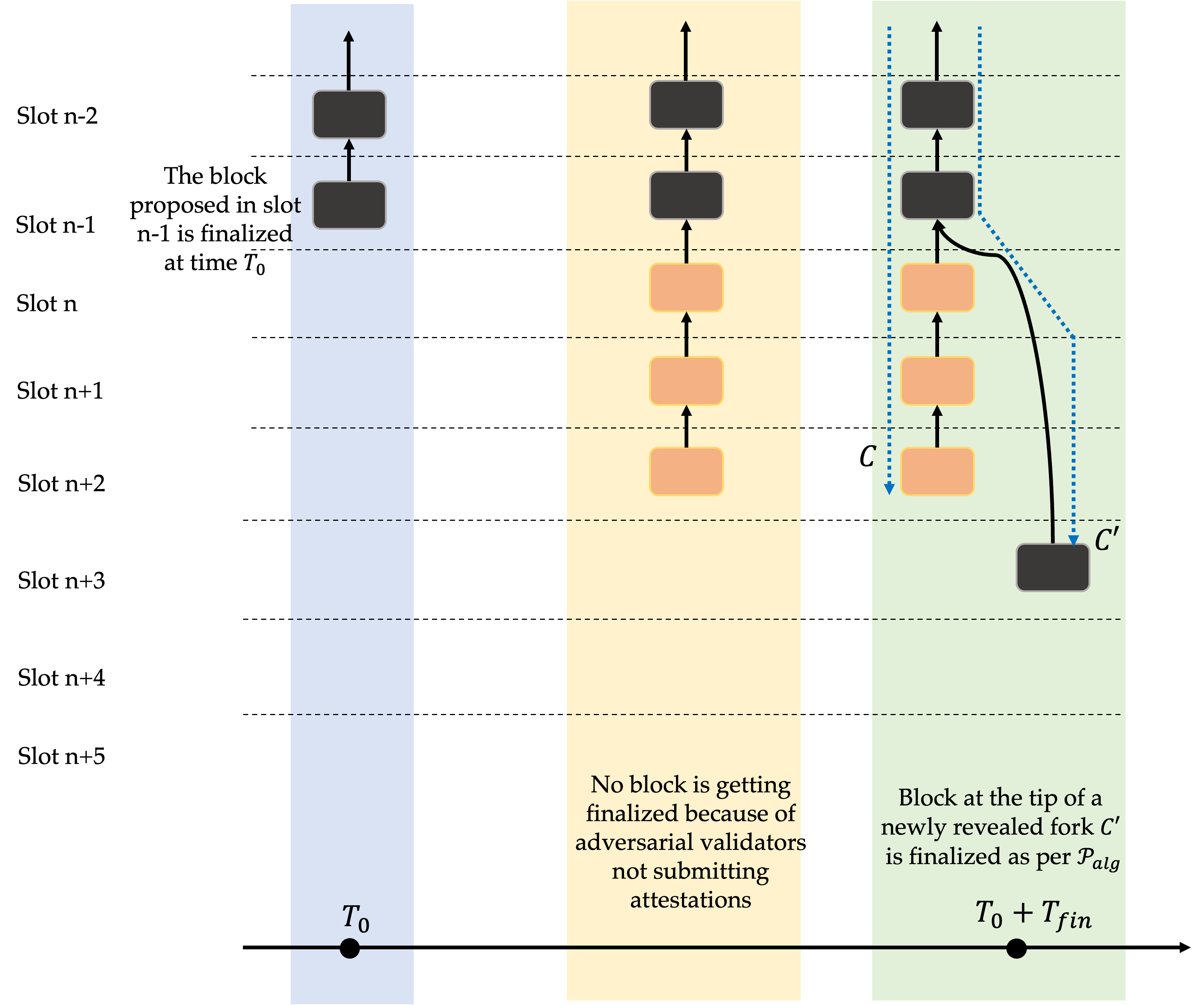}
         \caption{With more than $\frac{1}{3}$ of total stake being adversarial, the adversary can withheld their attestations.}
         \label{fig:y equals x}
     \end{subfigure}
     \hfill
     \begin{subfigure}[b]{0.75\textwidth}
         \centering
         \includegraphics[width=\textwidth]{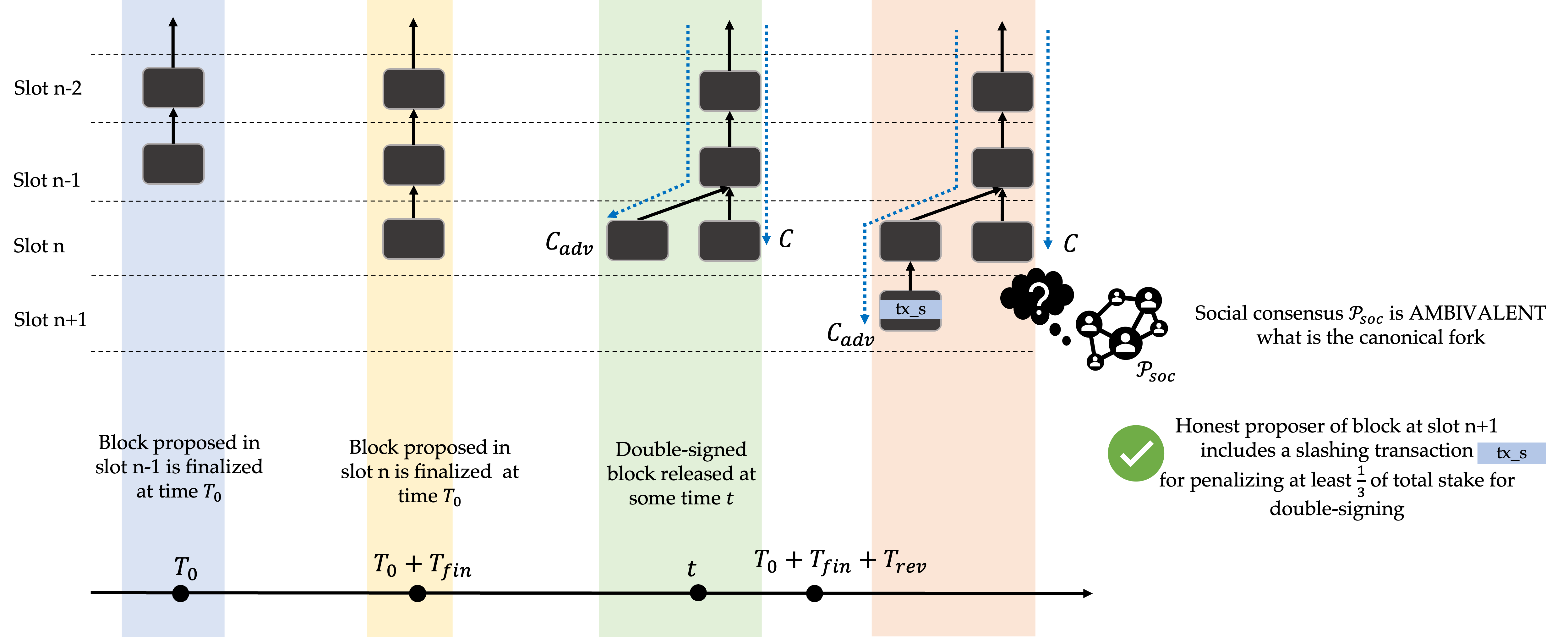}
         \caption{Given $\mathcal{C}_{adv}$ is revealed before $T_0 + \reversion$, honest participants in $\psoc$ are ambivalent which fork is the canonical one. }
         \label{fig:three sin x}
     \end{subfigure}
\end{figure}

\begin{figure}[ht]\ContinuedFloat
     \centering
     \begin{subfigure}[b]{0.75\textwidth}
         \centering
         \includegraphics[width=\textwidth]{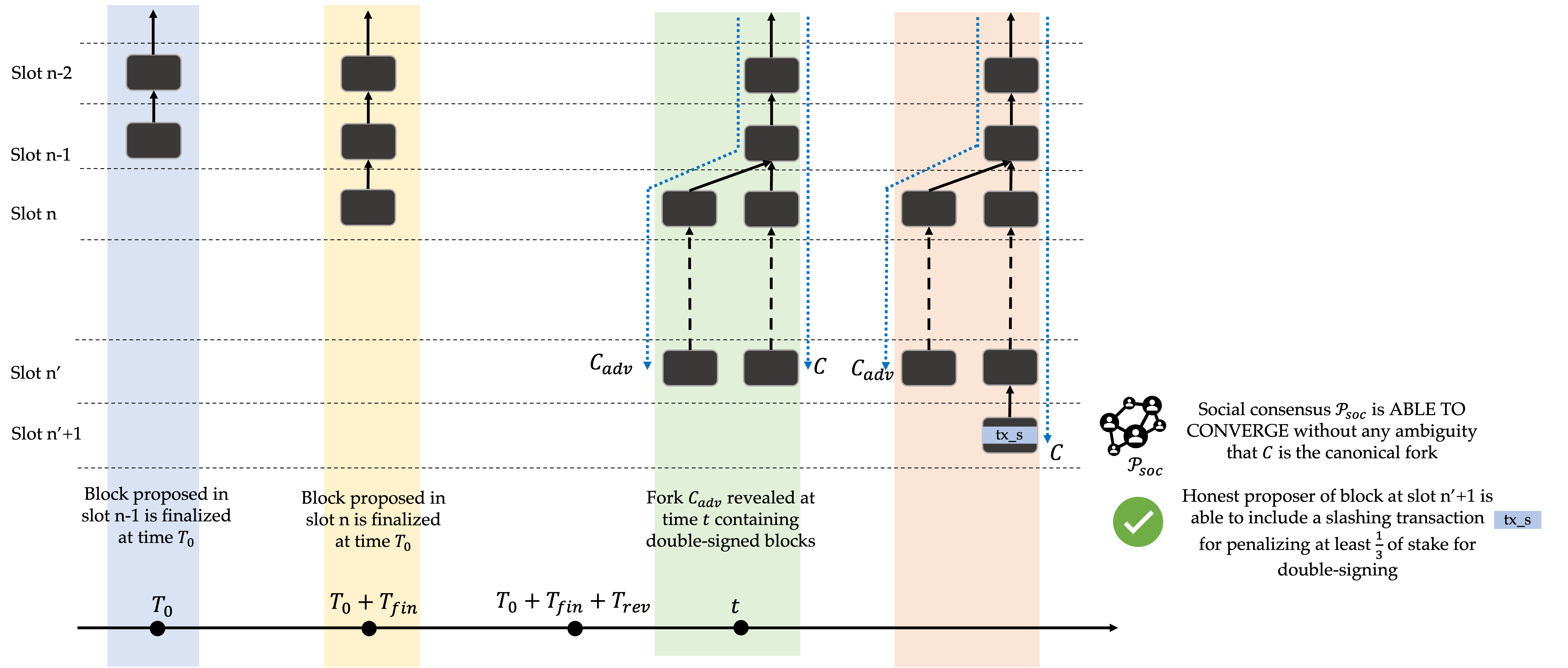}
         \caption{Given $\mathcal{C}_{adv}$ is revealed after $T_0 + \reversion$, honest participants in $\psoc$ are able to converge without any ambiguity which fork is the canonical one. }
         \label{fig:reversion-period}
     \end{subfigure}
     \hfill
     \begin{subfigure}[b]{0.75\textwidth}
         \centering
         \includegraphics[width=\textwidth]{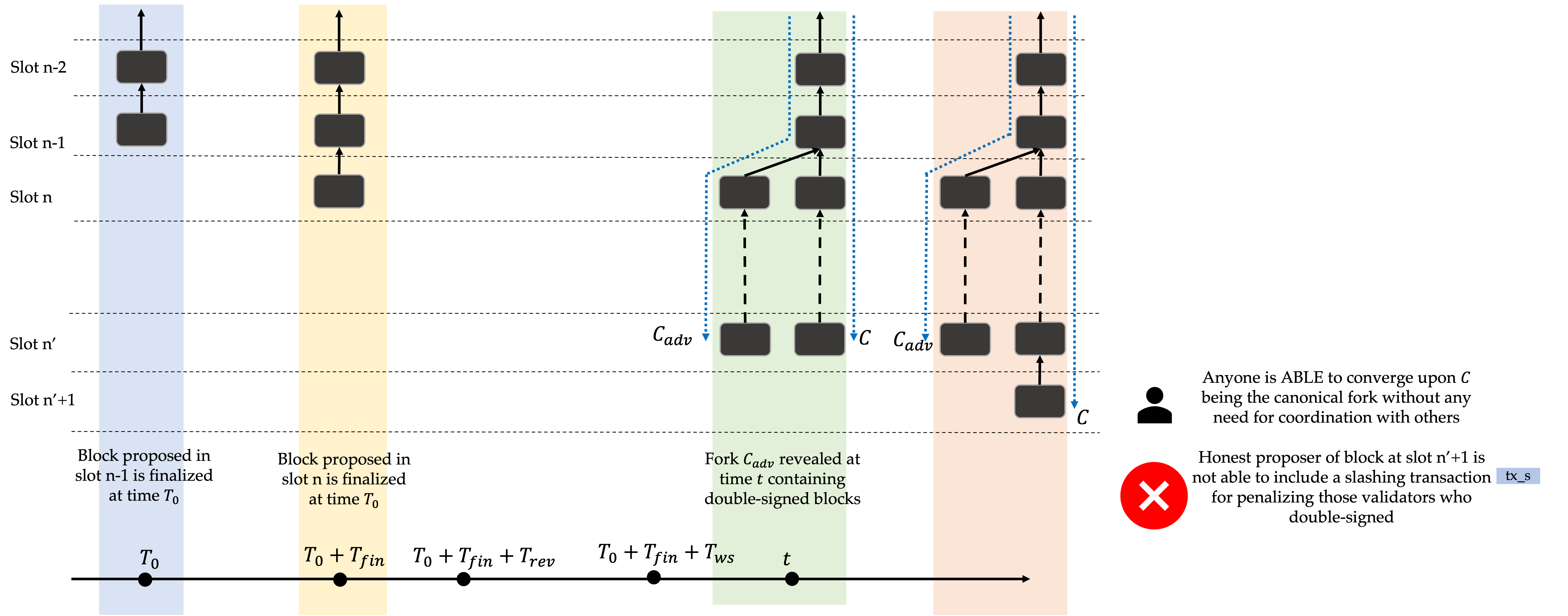}
         \caption{Given $\mathcal{C}_{adv}$ is revealed after $T_0 + T_{ws}$, honest participants in $\psoc$ are able to converge without any ambiguity which fork is the canonical one. However, computation of $T_{ws}$ is such that $\frac{2}{3}$ of the stake that is present in $\mathcal{C}_{adv}$ cannot be slashed in canonical fork as they have already exited in $C$ by the time double-signing is identified.}
         \label{fig:five over x}
     \end{subfigure}
     \caption{An illustration of different resolution periods.}
     \label{fig:three graphs}
\end{figure}

In theory, the total volume transacted within the reversion period could be even more than the TVL. In such a scenario, clearly the above security condition won’t hold, which is contrary to what is happening in practice. Therefore, we need a sharper reasoning on what constitutes profit-from-corruption.  

Chain reorgs particularly affect hybrid transactions where one person receives assets in physical or a separate digital platform from another person in exchange for sending digital assets in this blockchain. Hybrid transactions result in both on-chain and off-chain state changes. Examples of hybrid transactions are token-fiat exchange transactions in centralized exchanges (CEXs) and cross-chain transactions over bridges. What can happen is that the fork that ended up being considered to be “canonical” features the on-chain digital state caused by the above exchange to get reverted.  However, the  corresponding physical state doesn’t get reverted atomically. These hybrid transactions are in contrast to pure transactions where reorgs result in state changes that are limited to only on-chain state updates in the blockchain only. Any chain reorg in the blockchain reverts back all these state updates atomically. We represent the family of hybrid transactions as $\hybtxset$ and set of hybrid transactions from time $t_1$ to time $t_2$ as $\hybtxset(t_1, t_2)$.

Coming back to our estimation of profit-from-corruption, given the absence of the escape hatch of off-chain state changes getting reverted if a transaction belonging to set $\hybtxset$ gets reorged out, we should only consider the maximum value that has been transacted via transactions in the set $\hybtxset(t, t+\reversion)$ for any time $t>0$.  See Fig. 4 for an illustration. Hence, the security condition that needs to be satisfied for a PoS blockchain to be secure is that:
\begin{equation}
    \frac{1}{3} \times \stot > \sup_{t > 0} \sum_{tx \in \hybtxset (t,t+\reversion)} v(tx). 
\end{equation}
Currently, total volume of transactions of type $\hybtxset$ is dominated by CEXs and bridge liquidity. So, to determine a very sharp estimate of profit-from-corruption in PoS blockchains, we just need to determine the outflow of total value via CEXs and bridges within the reversion period.

\subsection{Reducing profit-from-corruption}
From previous section, we have the following security condition for when a PoS blockchain can be considered to be secure:
\begin{equation}
    \frac{1}{3} \times \stot > \sup_{t > 0} \sum_{tx \in \hybtxset (t,t+\reversion)} v(tx). 
\end{equation}
Clearly, to reduce the profit-from-corruption, two antecedents needs to be satisfied:  
\begin{itemize}
    \item $\reversion$ should be as close as possible to $0$,
    \item Total volume transacted via transactions of type $\hybtxset$ should be as minimal as possible.
\end{itemize}
In this section, we will investigate how to achieve both the above antecedents without doing any surgical modification to the underlying algorithmic consensus $\palg$ of a PoS blockchain.

\subsubsection{Get $\reversion$ to be close to $0$}
Recall that if there is a breakdown of the underlying honest majority assumption and an adversarial fork is publicly revealed to be finalized at time $t$, where $T_0 +T_{fin} \leq t < T_0+ T_{fin} + \reversion$, it is socially impossible for honest clients in the social consensus $\psoc$ to converge about which fork being “canonical”. The reason behind social impossibility arises because of a lack of a strong enough social coordination that can enable honest clients to come to a convergence as soon as possible. Ideally, we want  $\reversion$ to be as close to $0$ as possible.

Only known decentralized way to achieve a strong social coordination is via a network of clients running light clients that can:
\begin{itemize}
    \item monitor for block headers of finalized blocks,
    \item perform data availability sampling (DAS) of the blobs behind the block headers \cite{al2021fraud}, and 
    \item communicate with each other via some channel (say Twitter, Discord, etc.,) to share evidence, etc.
\end{itemize}
Having a lot of clients running light clients participate in monitoring for finalized block headers ensures that any progress in the canonical fork is detectable by all honest clients in the PoS blockchain (even wallets) as soon as possible. This ensures even if an adversarial fork is revealed to be finalized afterwards (owing to breakdown of honest majority assumption in the algorithmic consensus protocol), all honest clients in the social consensus $\psoc$ have already converged upon what should be considered the canonical fork. However, just detecting finalized block headers is not sufficient as the underlying blob for that block might not be available. Therefore, clients operating a light client software must also perform DAS. 

A social coordination consisting of a strong network of clients operating light clients ensures there is a very small interval of time within which a powerful adversary (that is, controls a fraction of stake beyond adversarial threshold) can confuse honest clients in social consensus $\psoc$ in coming to conflicting views on what should be the canonical fork. Therefore, stronger this social coordination is, more closer $\reversion$ is to $0$.

\subsubsection{Minimizing total value transacted via transactions belonging to the set $\hybtxset$}
\label{sec:sec-conf-hyb}
Even if $\reversion$ is close to $0$, there might be lot of value in the transactions belonging to the set $\hybtxset$ which are getting finalized within a time period of length $T_{fin}$. Typically with respect to transactions of type $\hybtxset$, clients employ the confirmation rule that as long as the block containing the transaction is finalized in the algorithmic consensus $\palg$, the transaction’s corresponding off-chain action will be executed. Naturally, what is necessary is the minimization of execution of off-chain actions corresponding to transactions of type $\hybtxset$ before the elapse of the reversion window of length $\reversion$ that started right after the client gets notified that the block containing the transaction has been finalized.

Keeping our above requirement in mind, we propose a new confirmation rule:
\begin{itemize}
    \item[] \textbf{\underline{Secure Confirmation Rule ($\secconf$)}}. Assume that the block $B$ containing a transaction $tx$ belonging to the set $\hybtxset$ is finalized. From the perspective of the client who has to take the corresponding off-chain action, suppose that no other adversarial fork is revealed to be finalized at a time $t$ that is within the reversion window of length $\reversion$ that started right after the block $B$ containing transaction $tx$ was finalized. Then, the client can execute the corresponding off-chain action. 
\end{itemize}

We can now make the following claim.
\begin{lemma}
    Confirmation rule $\secconf$ is unconditionally secure.
\end{lemma}
\begin{proof}
    The proof is simple. Suppose the client doesn’t observe any adversarial fork within the reversion window. If the client followed the above confirmation rule $\secconf$ before executing the off-chain action corresponding to the transaction belonging to the set $\hybtxset$, then it is guaranteed that all other honest clients in the blockchain’s ecosystem have unconditionally made up their mind that the fork containing this transaction is the canonical fork. Hence the on-chain state changes due to the transaction will never be reverted even if an adversarial fork is revealed to be finalized after the reversion window has passed. See fig.~\ref{fig:sec-conf} for an illustration. Hence, the secure confirmation rule $\secconf$ is unconditionally secure (or as called as $0$-0f-$N$ in \cite{buterin2022proof}). 
\end{proof}

\begin{figure}
    \centering
    \includegraphics[width=0.9\textwidth]{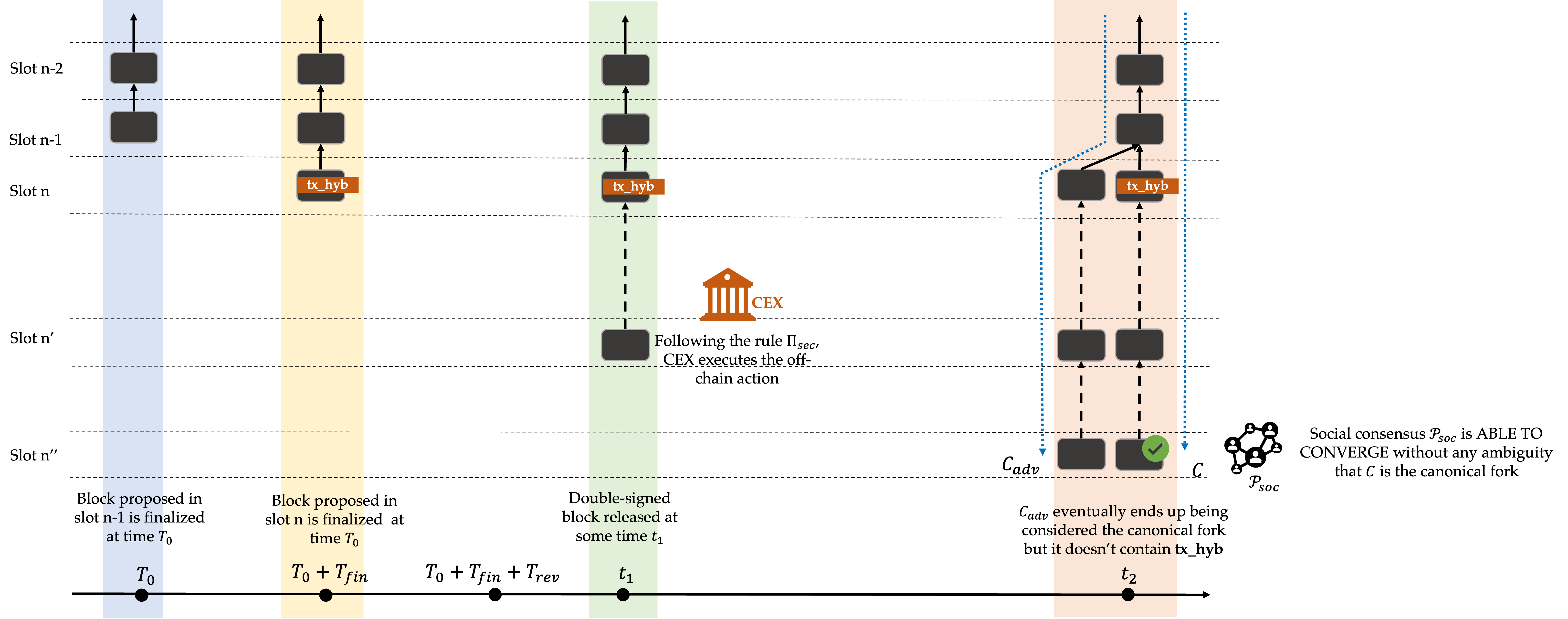}
    \caption{The secure confirmation rule $\secconf$ requires that the CEX executes the off-chain action corresponding to the transaction tx\_hyb after the reversion window has passed. This ensures that the fork that ends up being considered canonical by the social consensus $\psoc$ will have the transaction tx\_hyb.}
    \label{fig:sec-conf}
\end{figure}

Observe that the secure confirmation rule doesn’t require any changes to a PoS blockchain’s underlying algorithmic consensus $\palg$. This new confirmation rule is purely a modification to clients running light client software.

The implication of this secure confirmation rule is that any client following this rule will not end up in a scenario where it has executed the off-chain action but its corresponding on-chain state change got reverted. Let $\hybsecconftxset$ represent the set of transactions from the set $\hybtxset$ whose off-chain actions are being executed by clients by following the rule $\secconf$. Hence, with this secure confirmation rule $\secconf$, we have the following condition for when a PoS blockchain can be considered to be secure:
\begin{equation}
    \label{eq:sec-conf-hyb}
    \frac{1}{3} \times \stot > \sup_{t > 0} \sum_{tx \in \hybtxset (t,t+\reversion) \setminus \hybsecconftxset(t,t+\reversion)}  v(tx). 
\end{equation}

\subsubsection{Implementing secure confirmation rule}
As mentioned before, we have two major categories of clients that execute off-chain action emitting from transactions belonging to $\hybtxset$: (1) CEXs and (2) bridges. Let us see how they can implement this confirmation rule $\secconf$.
\newline

\noindent\underline{CEXs} \newline
It is straightforward for CEXs to follow the confirmation rule $\secconf$: CEXs will transfer the fiat money to a customer only after the reversion period for the corresponding on-chain deposit transaction has passed.
\newline

\noindent\underline{Bridges} \newline
We will consider light client bridges from a PoS blockchain A to another PoS blockchain B. Drawing analogy from how CEXs can implement the confirmation rule $\secconf$, one might say that light client smart contract in the receiving chain $B$ will accept a finalized block header $BH$ of chain $A$ if no conflicting finalized block header of chain $A$  is posted to chain $B$ within reversion window after the posting of block header $BH$ in chain $B$ (assuming $\mathcal{O}(1)-of-N$ trust assumption \cite{buterin2022proof}). 

Unfortunately, this doesn’t take into consideration the possibility of censorship in the receiving chain $B$. Assuming censorship-resistant window of the receiving chain $B$ to be $\crw$, what can happen is that the transaction containing the conflicting finalized block header of chain $A$ is censored in the receiving chain $B$ by a window of at max $\crw$. It is possible that $\crw > \reversion$. Clearly if the light client contract in the chain $B$ immediately confirms the block header $BH$ after the reversion window has passed, then we might end up in a situation where the corresponding on-chain state of chain $A$ has been reverted. 

Hence, the secure confirmation rule in case of bridges is:
\begin{itemize}
    \item[] \underline{\textbf{Secure Confirmation Rule for Bridges ($\secbri$)}}. After the transaction containing the block header $BH$ of a finalized block on chain $A$ is posted in the receiving chain $B$, the light client bridge smart contract in chain $B$ should wait for a window of length $\reversion + \crw$ before accepting $BH$ as the block header corresponding to a finalized block in the canonical fork. If no other block header for a finalized block from a conflicting fork of chain $A$ is posted by the end of this window in chain $B$, accept $BH$ as the block header corresponding to the finalized block in the canonical fork of chain $A$. 
\end{itemize}
It is apparent that in order for having better capital efficiency for bridges that are implementing the secure confirmation rule $\secbri$, it is important to have shorter censorship-resistant window $\crw$. 
\section{STAKESURE: A new economic mechanism for Staking insurance}
\label{sec:insurance}
Let us identify $\setclient$ as the exhaustive set of all transactors that are responsible for executing off-chain action corresponding to some transaction in the set $\hybtxset$. For any transactor $\client \in \setclient$, let $\clienthybtxset(t,t+\reversion)$ represents the set of transactions in $\hybtxset(t,t+\reversion)$ whose off-chain action is being executed by the transactor $\client$ without following the rule $\secconf$. Then, we have
\begin{equation}
    \hybtxset (t,t+\reversion) \setminus \hybsecconftxset(t,t+\reversion) = \bigcup_{\client \in \setclient}\clienthybtxset(t,t+\reversion)
\end{equation}
Now, we can rewrite the cryptoeconomic safety condition from Eq.~\ref{eq:sec-conf-hyb} as
\begin{equation}
    \frac{1}{3} \times \stot > \sup_{t > 0} \sum_{\client \in \setclient} \left(\sum_{tx \in \clienthybtxset(t,t+\reversion)}v(tx)\right)
\end{equation}
Unfortunately, in a decentralized system, evaluating $\left(\sum_{tx \in \clienthybtxset(t,t+\reversion)}v(tx)\right)$ for any transactor $\client$ is dependent on the transactor honestly disclosing off-chain actions of how many transactions from the set $\hybtxset(t,t+\reversion)$ did it execute without following the rule $\secconf$. But without credible commitment, the transactor can easily indulge in obfuscating this information to create over leverage. Therefore, we need an information signalling mechanism so that profit-from-corruption can be precisely evaluated.

We observe that the definition of cryptoeconomic safety does not really guarantee that a transaction user enjoys unconditional safety, rather it only says that an attacker does not derive profit from the attack. However, in complex scenarios, it is possible that an attacker may attack out of pure malice or other reasons, and a honest transactor is affected. We therefore define a stronger notion of cryptoeconomic safety as follows:
\begin{itemize}
    \item[] \underline{\textbf{Strong cryptoeconomic safety.}} No honest user of the system suffers any loss of funds. 
\end{itemize}
As noted in the Introduction, we remind that this is a much stronger definition than the definition of cryptoeconomic safety. While cryptoeconomic safety ensures that there is no incentive for an adversary to attack, a malicious adversary may still go ahead and attack the system which will lead to honest users in the system suffering without recourse. In contrast, in a system with strong cryptoeconomic safety, this can never happen.

In this section, we design and show a mechanism, called as STAKESURE, which achieves this stringent property while also solving for the information signalling problem from the previous section.

\subsection{Key ideas}
We first highlight the key ideas required to construct the mechanism before delving into details of the mechanism. 

\subsubsection*{1: What to do with the slashed funds?: Allocate to insurance!}
In existing PoS blockchains with slashing, the slashed funds are simply burnt. Instead slashed funds can be used to compensate for the harmed victims. This is the key principle that we build upon to obtain our requisite mechanism. 

\subsubsection*{2: Who should we allocate the slashed funds to?: Let users self-buy insurance!}
\begin{enumerate}
    \item There is an auction mechanism which runs on chain, which is used by transactors to buy up insurance for the upcomign period. 
    \item How much insurance to buy? Each transactor should buy up enough insurance to cover their damage in case that there is a reorg attack. Thus the insurance mechanism provides a way for transactors to convey their truthful value, otherwise, in case of an attack, they will be left without coverage.
    \item We note that transactors who have purely on-chain atomic transactions, may not care about reversion as in any valid state, they still do not suffer any significant loss. Thus transactions not in the set $\hybtxset$ are likely to not need any insurance.
    \item We note that transactors who use the secure confirmation rule $\secconf$ also do not need insurance, as they wait for their transaction to be confirmed beyond doubt of reversion.
\end{enumerate}

\subsubsection*{3: When to sell insurance?: Early enough so that insurance itself cannot be reverted!}
\begin{enumerate}
    \item In a system where attacks can happen, even the insurance transactions may be reverted. So we need a mechanism to wherein every epoch (which is set equal to a reversion period), there is an auction mechanisms which auctions out the slashed funds if slashing happens in the epoch two epochs ahead. This pre-allocation of insurance is important as it ensures that the holder of the insurance has a strong guarantee that the insurance transaction itself cannot be reverted. This principle also clarifies why the epoch is set equal to the reversion period.
    \item If epoch $e:=[t,t+\reversion)$, the amount of insurance that a transactor $\client$ would buy is based on its expectation on how much value would flow via transactions belonging to the set $\clienthybtxset(t,t+\reversion)$. Hence, the profit-from-corruption in epoch $e$ is precisely given by $\sum_{\client \in \setclient} \left(\sum_{tx \in \clienthybtxset(t,t+\reversion)}v(tx)\right)$. 
\end{enumerate}

\subsubsection*{4: Rational transactor policy.}
It is in the economic interest of a rational transactor to obtain enough insurance so that the value of the insurance they hold is greater than the funds at risk. We can expect that transactors with significant exposure will cover their risk by holding enough insurance. For example, exchanges, and bridges.

\subsubsection*{5: Self-scaling security.} 
Since rational transactors only transact if they have enough coverage, automatically the total cryptoconomic load on the system will be smaller than the total insurance coverage available, which is $\frac{\stot}{3}$. Thus even if only a smaller amount of stake is in the system, the system remains completely unconditionally safe. It is only the liveness of the honest transactors that get affected, i.e., they may have to wait to obtain insurance in order to transact. But this increases the insurance rate that such transactors may be willing to pay, in turn, increasing the total amount staked (as there is now enough return available for more stake). Thus staking in the system sets itself automatically to the right level of security rather than being controlled by an arbitrary preallocated rewards curve

\subsubsection*{6: Two pending issues}
\begin{enumerate}
    \item Small, irrational transactors.
    \begin{enumerate}
        \item However, it is possible that smaller transactors may not have the foresight to buy insurance or may simply risk their funds (trying to freeride on the assumed safety of the system). We need to make sure that there is enough cryptoeconomic buffer in the system for these transactors to exist. 
        \item We define $\hybuninstxset(t, t+\reversion)$ to be the set of transactions in $\hybtxset (t,t+\reversion) \setminus \hybsecconftxset(t,t+\reversion)$ for which the transactor hasn't brought insurance. Then, the total uninsured cryptoeconomic load is given by  
        \begin{equation*}
            \sup_{t > 0} \sum_{tx \in \hybuninstxset(t, t+\reversion)}  v(tx).
        \end{equation*} 
        \item We need to make sure that there is enough cost-of-corruption to protect against these small transactors, even though they do not have any insurance.
    \end{enumerate}
    \item Grieving attacks.
    \begin{enumerate}
        \item It is possible for an adversary to be the majority staker as well as buyer of all insurance, so that the adversary prevents any honest user from transacting in the system (as there is no insurance protection for them).
    \end{enumerate} 
\end{enumerate}

\subsubsection*{7: Resolving the issues}
    \begin{enumerate}
        \item The way to solve both these problems is to allocate a certain amount of slashed fund to be purely burnt rather than to be allocated to insurance. This ensures some non-zero cost on grieving attacks. Furthermore, if the amount allocated to be burnt is greater than the benefit from attacking the  small / irrational transactors, then there is no incentive to attack them. 
    \end{enumerate}

\subsection{Protocol description}
In this section, we will put together the ideas from the previous section to describe the STAKESURE mechanism.

\subsubsection*{Staking and funds for insurance coverage }
Whenever a new validator joins the PoS blockchain by staking a stake $\stake$ in the underlying consensus protocol, $\gamma$ fraction of this stake is earmarked for underwriting the risk under the proposed insurance mechanism. In return, the validator receives an insurance premium from anyone who buys this insurance. Note that here $\gamma$ is a system parameter. Now, this validator has three revenue streams: (1) reward issuance earned from staking a stake $S$, (2) MEV extractions from manipulation of transaction ordering, (3) premium earned from providing an insurance coverage of  worth $\gamma S$.

\subsubsection*{Optimistic path}
We partition the time horizon into epochs of length $\reversion$. Let $\sins(e)$ denotes the stake that is available in epoch $e$ for purchase by transactors in order to get insurance in epoch $e+2$. For epoch $e$, we set $\sinsavail(e) := \min\left\{\sins(e), \frac{\gamma\stot}{3}\right\}$. In epoch $e$, an auction will be conducted wherein any transactor can make bids for purchasing insurance coverage from this available stake $\sinsavail(e)$. If a transactor is able to win the bid for a stake of size $s$, this means that the transactor is guaranteed an insurance coverage for value transacted amounting to worth $s$. More explicitly, if the transactor is optimizing for fast UX by not following the secure confirmation rule $\secconf$ in executing the off-chain action corresponding to a transaction belonging to the set $\hybtxset$, then the transactor can bid for the insurance. If the transactor wins the bid, it can use the insurance for underwriting the risk of prematurely executing the off-chain action.

A natural question is when is the transactor permitted to buy insurance coverage for its transaction that the transactor expects to be executed on-chain in the epoch $e+2$. Observe that the transactor can execute the corresponding off-chain action immediately after the transaction has been finalized on-chain only if it is assured by the beginning of epoch $e$ that a social consensus has been achieved on a canonical fork which confirms its purchase of insurance. This means that the transactor must have bought the insurance at least one reversion period $\reversion$ before the beginning of the epoch $e$, that is, by end of epoch $e-2$. Therefore, insurance purchase rule is as follows:
\begin{itemize}
    \item[] \underline{\textbf{Insurance purchase rule ($\inspur$)}}.  For insurance coverage in epoch $e+2$, a transactor must make a bid from out of the available stake $\sinsavail(e)$  at epoch  $e$.  
\end{itemize}

Let $\clientinspur(\client, e)$ denote the total amount of insurance that has been bought by a transactor $\client$ in epoch $e$.  The transaction confirmation rule that the clients would follow in epoch $e+2$ is:
\begin{itemize}
    \item[] \underline{\textbf{Confirmation rule with insurance ($\inssec$)}}. Suppose by the beginning of epoch $e+2$, the transaction associated with the purchase of insurance by clients in epoch $e$ is confirmed to be part of the fork that is considered to be “canonical” by the social consensus. Let $\Gamma_{\client, e+2}$ be the set of transactions belonging to $\hybtxset$ whose off-chain action has been executed immediately by the transactor $\client$ right after they were finalized in epoch $e+2$. The condition that the transactor must ensure is that 
    \begin{align*}
        \sum_{tx \in \Gamma_{\client, e+2}} v(tx) < \clientinspur(\client, e).
    \end{align*}
\end{itemize}

For any transaction in the set $\hybtxset$ that is finalized by end of epoch $e+2$ and if there is no double-signing attack by the end of epoch $e+3$ (equivalently, within one reversion period  $\reversion$), then that transaction is recognized to be included in the fork which is considered to be “canonical” by the social consensus. Consequently, the stake being used for providing insurance coverage to the transactors is released and is up for auction at the beginning of epoch $e+4$. See fig.~\ref{fig:insurance} for an illustration for this insurance process.

\begin{figure}
    \centering
    \includegraphics[width=0.9\textwidth]{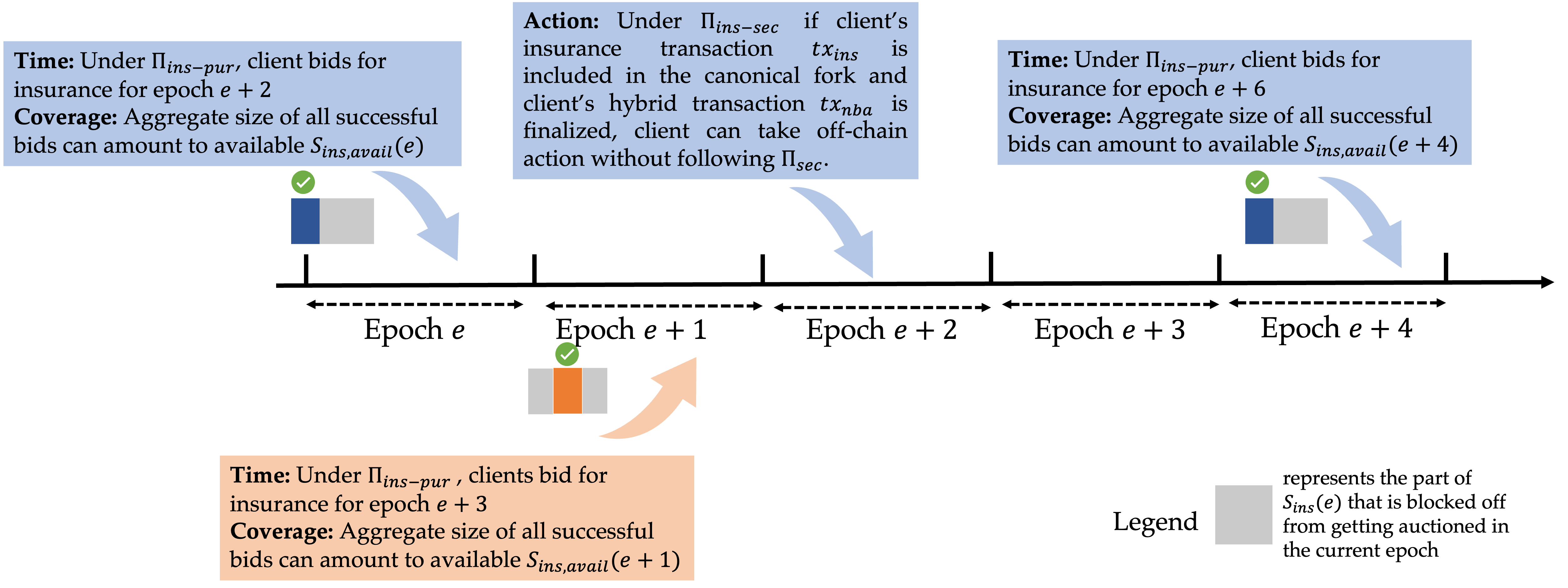}
    \caption{An illustration of the insurance mechanism under optimistic path.}
    \label{fig:insurance}
\end{figure}

\subsubsection*{Pessimistic path}
Following the confirmation rule $\inssec$, if by the beginning of epoch $e+2$, the transactor is not able to win the insurance bid (due to censorship or lack of enough insurance coverage), then the transactor is recommended to follow the secure confirmation rule $\secconf$ for executing the off-chain action corresponding to its transaction belonging to the set $\hybtxset$ that got finalized on-chain in epoch $e+2$.

Suppose that before the end of epoch $e+3$, the system is attacked by the adversary via double signing on finalized blocks. As soon as double-signing is detected, the transactors are recommended to follow the secure confirmation rule $\secconf$ for all future epochs until a social consensus is achieved that the attack is over. Eventually by some future epoch $e'$, the attack will be over and social consensus will converge upon what should be considered to be the canonical fork. From epoch $e'$, any transactor who doesn’t plan to follow the secure confirmation rule $\secconf$ rule but expect to have a transaction of type $\hybtxset$ getting executed on-chain in epoch $e'+2$ will bid for buying insurance from those stake in $\sinsavail(e')$ which are available for bidding.

We suggest a small modification to the slashing rule that is nominally used for penalizing double signing in PoS blockchains like Ethereum and Tendermint:
\begin{itemize}
    \item[] \underline{\textbf{Penalty for double signing ($\Pi_{ds}$)}}. For penalizing the validator (with stake $\stake$) who participated in double signing,  $\gamma$ fraction of the stake is used for paying out the insurance to the transactors who bought insurance from the validator. Remaining $1-\gamma$ fraction is burnt just like in usual slashing for double signing.  
\end{itemize}

\subsection{Consequences of STAKESURE}
Some major consenquences of the insurance scheme are illustrated here.
\begin{enumerate}
    \item Perfectly secure bridging across different chains is possible!  
    \begin{itemize}
        \item In a recent article by V. Buterin \cite{vitalikbridge}, it was highlighted that bridging is prone to the fundamental risk that the two domains (source domain and destination domain) can move out of sync - one concrete context where this can happen is if there is a $51$\% attack in the sending side of the bridge. At any moment, the bridge has to purchase enough long-term insurance under STAKESURE that is enough to secure against all the tokens in flight at that moment. This quantifies the economic cost in terms of premium paid for subscribing to STAKESURE for building a cryptoeconomically safe bridge. What we have presented is an initial concept on building bridges with strong cryptoeconomic safety and further research is needed to flush out the complete design.
    \end{itemize}
    \item Insured light clients.
    \begin{itemize}
        \item Since checking whether a client has enough insurance is possible to do via a light node, we can have insured light clients, i.e., nodes with very little computational abilities that can operate with very high security (in fact, strong cryptoeconomic safety). 
    \end{itemize}
    
    \item Self-scaling security: System works even without emissions.
    \begin{itemize}
        \item Existing blockchains, including Ethereum, design a staking rewards which allocates a certain amount of new tokens to be printed and given as rewards. This is to ensure that the system has enough cryptoeconomic security. STAKESURE ensures that the system can automatically find out how much cryptoeconomic security is needed by looking at how much insurance is needed and allocate it. If the total amount staked is not sufficient to satisfy the insurance demand, then automatically the price of insurance increases, increases the staking returns, which incentivizes new stake to enter the system
    \end{itemize}

    \item Pricing the right dimension
    \begin{itemize}
        \item Existing fee mechanisms depend on the assumption that the demand of computation of the system exceeds the computational capacity of the system, in order to form any non-zero price. In a blockchain with infinite computational resources (or even resources far exceeding demand), the price will be very small. In our proposal, we propose to separate the insurance price from the congestion price (which is paid for computational congestion), and even when congestion price goes to zero, the insurance price in STAKESURE remains non-zero, as long as there is demand for computation. Furthermore, imagine a simple transfer transaction worth billions of dollars, this may be very simple computationally, but imposes too much cryptoeconomic risks on the system, so it is only appropriate that the transactor pays for the cryptoeconomic risks created. 
    \end{itemize}

 \item {Relative volatility as a key consideration of the denomination of insurance}
    \begin{itemize}
        \item It is ideal if the asset that is being covered by insurance be in the same denomination as that of the denomination in which insurance claim will be paid. For example, if a buyer buys insurance against loss of ETH, then the denomination in which insurance claims are paid are also in ETH. If they are not in the same denomination, then future volatility between different denominations has to be considered in obtaining the insurance. Thus as more blockchain-native markets price in ETH, it is much better for ETH to be the insuring asset, rather than some other token of value. The same effect holds in the other way as well: when most assets are priced in USD stablecoins, then having USD staked may serve better in terms of insurance cryptoeconomics (there may be other effects such as sovereignty and forkability that need to be considered). 
    \end{itemize}

\end{enumerate}

\subsection{Comparison with traditional insurance}
Relative to traditional insurance mechanisms, there is a variety of differences that enable this proposed insurance mechanism to operate with a high degree of reliability. 
\begin{enumerate}
    \item \textbf{Not statistical.} Traditional insurance schemes make use of statistical multiplexing when deciding how much fund they should keep in reserve for meeting insurance claims at any point in time. For example, car insurance companies make a statistical estimate about how many car accidents will happen among the insurance buyers. Depending on that estimate, car insurance companies keep in reserve sufficient funds to meet the  insurance claims that will be made in a year. This leads to the insurance providers making over-promise on how much claims they can compensate for and are vulnerable to risk of collapse during black swan events, especially when there are correlated failures. Our proposed insurance mechanism gives a deterministic guarantee that even if all transactors who have subscribed to insurance get harmed, there is enough funds to pay out the compensation to all those transactors.
    \item \textbf{A closed loop of Karma.} Traditional insurance is not a closed system. What this means is that the person underwriting the risk via insurance is not the person who is causing the harm (for example, in the case of motor accident insurance, the accident is caused by someone else but insured by the insurance agency). In case of PoS blockchains where clients are not following the secure confirmation rule $\secconf$, the entities who can cause harm are the validators who are participating in the underlying algorithmic consensus. In our proposed insurance mechanism, validators staked in the algorithmic consensus $\palg$ are the entities who offer insurance, ensuring that there is a closed loop of Karma.
    \item \textbf{No moral hazard.} The proposed insurance mechanism underwrites risk for only safety violations where the claim that harm has been inflicted on the insured is objectively provable. So there is no moral hazard in the insurance, where a party claims insurance by falsely claiming to be hurt (as they will never be able to provide a certificate of violation by another party). 
    \item \textbf{Grieving imposes penalties.} Burning $(1-\gamma)$ fraction of slashed funds ensures that an adversary cannot grieve in our proposed insurance mechanism by being both the staker and the insurance buyer. In such a case, there is a guaranteed  
\end{enumerate}

\section{Conclusion}
In this work, we have formally introduced an economic model for analyzing the cryptoeconomic safety of a PoS blockchain. This involved computing minimum cost to corrupt a blockchains and comparing it with the maximum profit that can be extracted from corrupting a blockchain. We showed that incorporating slashing in the PoS blockchain enhances the cost-of-corruption. We also derived sharper bounds on  profit-from-corruption which required identifying the set of transactions that are impacted by chain reorgs. Finally, we introduced an on-chain insurance mechanism, called as STAKESURE, that guarantees strong cryptoeconomic safety: no honest user of the system suffers any loss of funds. 

\bibliographystyle{unsrt}
\bibliography{references}

\begin{thebibliography}{10}

\bibitem{buterin2022proof}
Vitalik Buterin.
\newblock Proof ofstake.
\newblock 2022.

\bibitem{ultrasound}
{ultrasound.money}, Accessed 2023.
\newblock \url{https://ultrasound.money/}.

\bibitem{nas}
{Nothing-at-stake problem}, Accessed 2023.
\newblock \url{https://golden.com/wiki/Nothing-at-stake_problem-639PVZA}.

\bibitem{dembo2020everything}
Amir Dembo, Sreeram Kannan, Ertem~Nusret Tas, David Tse, Pramod Viswanath,
  Xuechao Wang, and Ofer Zeitouni.
\newblock Everything is a race and nakamoto always wins.
\newblock In {\em Proceedings of the 2020 ACM SIGSAC Conference on Computer and
  Communications Security}, pages 859--878, 2020.

\bibitem{deb2021posat}
Soubhik Deb, Sreeram Kannan, and David Tse.
\newblock Posat: proof-of-work availability and unpredictability, without the
  work.
\newblock In {\em Financial Cryptography and Data Security: 25th International
  Conference, FC 2021, Virtual Event, March 1--5, 2021, Revised Selected
  Papers, Part II 25}, pages 104--128. Springer, 2021.

\bibitem{yin2018hotstuff}
Maofan Yin, Dahlia Malkhi, Michael~K Reiter, Guy~Golan Gueta, and Ittai
  Abraham.
\newblock Hotstuff: Bft consensus in the lens of blockchain.
\newblock {\em arXiv preprint arXiv:1803.05069}, 2018.

\bibitem{chan2020streamlet}
Benjamin~Y Chan and Elaine Shi.
\newblock Streamlet: Textbook streamlined blockchains.
\newblock In {\em Proceedings of the 2nd ACM Conference on Advances in
  Financial Technologies}, pages 1--11, 2020.

\bibitem{buchman2018latest}
Ethan Buchman, Jae Kwon, and Zarko Milosevic.
\newblock The latest gossip on bft consensus.
\newblock {\em arXiv preprint arXiv:1807.04938}, 2018.

\bibitem{Augur}
Jack Peterson, Joseph Krug, Micah Zoltu, Austin~K. Williams, and Stephanie
  Alexander.
\newblock Augur: a decentralized oracle and prediction market platform.
\newblock 2018.

\bibitem{kleros}
{Kleros}, Accessed 2023.
\newblock \url{https://kleros.io/whitepaper.pdf}.

\bibitem{pplusepsilon}
{P+epslion attack}, Accessed 2023.
\newblock \url{https://blog.ethereum.org/2015/01/28/p-epsilon-attack}.

\bibitem{weaksubjectivity}
{Weak subjectivity attack}, Accessed 2023.
\newblock
  \url{https://blog.ethereum.org/2014/11/25/proof-stake-learned-love-weak-subjectivity}.

\bibitem{weaksubjectivitycalc}
{Analysis of Weak subjectivity in Ethereum 2.0}, Accessed 2023.
\newblock
  \url{https://github.com/runtimeverification/beacon-chain-verification/blob/master/weak-subjectivity/weak-subjectivity-analysis.pdf}.

\bibitem{vitaliksocialconsensus}
{Responding to 51\% attacks in Casper FFG}, Accessed 2023.
\newblock
  \url{https://ethresear.ch/t/responding-to-51-attacks-in-casper-ffg/6363}.

\bibitem{al2021fraud}
Mustafa Al-Bassam, Alberto Sonnino, Vitalik Buterin, and Ismail Khoffi.
\newblock Fraud and data availability proofs: Detecting invalid blocks in light
  clients.
\newblock In {\em Financial Cryptography and Data Security: 25th International
  Conference, FC 2021, Virtual Event, March 1--5, 2021, Revised Selected
  Papers, Part II 25}, pages 279--298. Springer, 2021.

\bibitem{vitalikbridge}
Accessed 2023.
\newblock
  \url{https://old.reddit.com/r/ethereum/comments/rwojtk/ama_we_are_the_efs_research_team_pt_7_07_january/hrngyk8/}.

\end{thebibliography}
\end{document}